\documentclass[conference]{IEEEtran}
\usepackage[papersize={8.5in, 11in}, left=1.6cm, right=1.6cm, top=0.75in, bottom=1.02in]{geometry}

\IEEEoverridecommandlockouts
\usepackage{amsmath,amssymb,amsfonts}
\usepackage{amsthm}
\usepackage{comment}

\usepackage{graphicx}
\usepackage{textcomp}
\usepackage{xcolor}
\def\BibTeX{{\rm B\kern-.05em{\sc i\kern-.025em b}\kern-.08em
    T\kern-.1667em\lower.7ex\hbox{E}\kern-.125emX}}

\graphicspath{ {./figures/} }

\usepackage[numbers, sort&compress]{natbib}
\usepackage[hidelinks]{hyperref}

\usepackage{tabularx}
\usepackage{booktabs}
\usepackage{cleveref}
\usepackage{subcaption}
\usepackage{multirow}

\usepackage{enumitem}
\usepackage{tikz}
\usepackage{circledsteps}
\newcommand{\circled}[1]{\tikz[baseline=(char.base)]{\node[shape=circle,draw,inner sep=1.0pt](char){#1}}}

\usepackage{algorithm}
\usepackage{algpseudocode}
\algrenewcommand\algorithmicrequire{\textbf{Input:}}
\algrenewcommand\algorithmicensure{\textbf{Output:}}

\algnewcommand{\LeftComment}[1]{\Statex \(\triangleright\) #1}

\newtheorem{lemma}{Lemma}

\newtheorem{proposition}{Proposition}

\crefname{definition}{Definition}{Definitions}
\crefname{lemma}{Lemma}{Lemmas}
\crefname{remark}{Remark}{Remarks}
\crefname{proposition}{Proposition}{Propositions}
\crefname{corollary}{Corollary}{Corollaries}
\crefname{algorithm}{Algorithm}{Algorithms}

\makeatletter
\newcommand{\multiline}[1]{%
  \begin{tabularx}{\dimexpr\linewidth-\ALG@thistlm}[t]{@{}X@{}}
    #1
  \end{tabularx}
}
\makeatother

\DeclareMathOperator*{\argmin}{arg\,min}

\newcommand{\sN}{\mathcal{N}}

\newcommand{\sV}{\mathcal{V}}
\newcommand{\sE}{\mathcal{E}}
\newcommand{\sG}{\mathcal{G}}
\newcommand{\sI}{\mathcal{I}}

\newcommand{\sO}{\mathcal{O}}

\newcommand{\sA}{\mathcal{A}}
\newcommand{\sR}{\mathcal{R}}

\newcommand{\mA}{\mathbf{A}}
\newcommand{\vw}{\mathbf{w}}

\begin{document}

\title{
Delay and Overhead Efficient Transmission Scheduling for Federated Learning in UAV Swarms

% \thanks{This work has been submitted to the IEEE for possible publication. Copyright may be transferred without notice, after which this version may no longer be accessible.}
% \thanks{D.~N.~M.~Hoang, V.~T.~Truong, H.~D.~Le, and L.~B.~Le are with INRS, University of Qu\'{e}bec, Montr\'{e}al, QC H5A 1K6, Canada. Emails: \{duc.hoang, tuan.vu.truong, duy.hung.le, long.le\}@inrs.ca.}
\thanks{The authors are with INRS, University of Qu\'{e}bec, Montr\'{e}al, QC H5A 1K6, Canada. Emails: \{duc.hoang, tuan.vu.truong, duy.hung.le, long.le\}@inrs.ca.}
% \thanks{}
\vspace*{-8pt}
}

% {\footnotesize \textsuperscript{*}Note: Sub-titles are not captured in Xplore and should not be used}

\author{\IEEEauthorblockN{Duc~N.~M.~Hoang, Vu~Tuan~Truong, Hung~Duy~Le, and Long~Bao~Le}
\vspace{-15mm}
}

\maketitle

\begin{abstract}
This paper studies the wireless scheduling design to coordinate the transmissions of (local) model parameters of federated learning (FL) for a swarm of unmanned aerial vehicles (UAVs).
The overall goal of the proposed design is to realize the FL training and aggregation processes with a central aggregator exploiting the sensory data collected by the UAVs but it considers the multi-hop wireless network formed by the UAVs.
Such transmissions of model parameters over the UAV-based wireless network potentially cause large transmission delays and overhead. 
Our proposed framework smartly aggregates local model parameters trained by the UAVs while efficiently transmitting the underlying parameters to the central aggregator in each FL global round.
We theoretically show that the proposed scheme achieves minimal delay and communication overhead.
Extensive numerical experiments demonstrate the superiority of the proposed scheme compared to other baselines.
\end{abstract}

\begin{IEEEkeywords}
Communication efficiency, delay efficiency, federated learning, model aggregation, unmanned aerial vehicles
\end{IEEEkeywords}

\section{Introduction} \label{sec:intro}
Unmanned aerial vehicles (UAVs) have emerged as a transformative technology, particularly for crowdsensing tasks in remote areas.
The UAVs can revolutionize data collection by using onboard sensors and cameras to gather a diverse range of information in regions where traditional methods face challenges due to geographical constraints or limited resources.
What enhances their utility is the integration of federated learning (FL)~\cite{mcmahan_communication-efficient_2017}, a distributed learning approach that harnesses the computational power of UAVs for in-the-air collaborative learning.
A swarm of UAVs can be rapidly deployed for FL-empowered crowdsensing missions~\cite{wang_learning_2021-1}, where they collect data and conduct model training thanks to their onboard computational resources and advanced sensors.
With the aid of FL, the machine learning (ML) model trained by the swarm can leverage the comprehensive dataset collected by all individual UAVs.
%is updated in each global aggregation step, resulting in significantly improved predictive performance.
Through this amalgamation of UAVs and FL, crowdsensing can be readily deployed in remote areas without requiring communications and computing infrastructure.
\par

To implement FL for a UAV swarm, one can require the UAVs to transmit their local model updates to a ground station or server for global model aggregation in each global round. However, in challenging scenarios like natural disasters, the existing stations can be compromised and thus the aggregation step in the FL training process can be affected. 
In addition, even the wireless transmission of model updates between the UAVs and the ground station in mission-critical applications (e.g., military surveillance operations) can pose risks of privacy leakage and eavesdropping~\cite{lu_differentially_2020} to sensitive information. 
In these infrastructure-less settings, the UAVs should collaboratively undertake data acquisition, model training, and model aggregation among themselves, eliminating the necessity for an external central station. To this end, an efficient transmission scheduling scheme for the model parameters to coordinate the FL training and aggregation among the UAVs is needed, which is the focus of this paper.
\par

Designing such a transmission scheduler for collaborative learning in an infrastructure-less UAV swarm, however, faces several challenges.
Firstly, traditional FL typically requires a predetermined central aggregator to collect and aggregate all the model updates submitted by local nodes in each global round.
For UAV swarms deployed in isolated and resource-deficient areas, there is a need for self-coordination among the UAVs to enable model aggregation without the help of an aggregator on the ground.
Secondly, traditional FL usually assumes that direct communications between local nodes and the aggregator are possible.
However, UAVs have a limited wireless communication range; therefore, the distance between two UAVs can be larger than the communication range, and multi-hop communications are thus necessary.
Supposing that a particular UAV is chosen to be the aggregator, multi-hop communications for exchanging the local model updates between it and other UAVs can introduce significant latency and transmission overhead.
%Finally, the transmission scheduling scheme should be designed taking into account communications delay and overhead. 
Thus, the development of an effective scheduling scheme supporting the model data exchanges and aggregation for FL in infrastructure-less UAV swarms is worth investigating.
\par

There are several studies on in-the-air FL~\cite{liu_federated_2021-1, wang_learning_2021-1, he_cgan-based_2023}. However, most of them assumed a centralized server for the deployed UAVs, which is not always available in scenarios without dedicated infrastructure. \citet{qu_decentralized_2021}~proposed an FL framework in which each UAV broadcasts its local models to immediate neighbors, and likewise, each UAV aggregates the model from its immediate neighbors. However, there was no detailed study of communication overhead and delay, and only limited numerical results were presented.
\citet{xiao_fully_2021}~proposed a learning scheme for FL in UAV swarms considering the varying locations of UAVs over time.
While the proposed scheme guarantees convergence and communication efficiency over dynamic scenarios, latency reduction is not investigated. \citet{qu_efficient_2021}~proposed a hierarchical FL framework to minimize the energy consumption during the training process in a UAV swarm, whose operation depends on a leading UAV.
However, communication overhead was not thoroughly studied.
\citet{feng_blockchain-empowered_2022}~devised a blockchain-empowered FL scheme for collaborative learning between the UAVs. However, the authors did not design a transmission scheduling scheme with low communication overhead and delay.
% \citet{al-abiad_decentralized_2023-1}~developed a decentralized FL framework in which local model updates are based on overlapped clustering to minimize energy usage. Although in their proposed design, model aggregation is not performed by a central aggregator, a centralized scheduling server is still required on the ground.
\par

Considering the aforementioned research gaps, this paper aims to develop an efficient transmission scheduling scheme for FL coordination in UAV swarm.
In detail, we first propose a transmission scheduler that realizes the conventional FL scenario by choosing a suitable aggregator UAV and scheduling the UAVs to efficiently transmit the model updates to the chosen aggregator, based on which we obtain a global model in each training round.
Next, by modeling the swarm network as a graph and partitioning time into discrete transmission time slots, we prove that the proposed scheme achieves minimum delay and communication overhead.
Extensive numerical experiments are conducted to show that the proposed scheme attains desired convergence and outperforms other baselines in terms of communication delay and overhead.

% The rest of this paper is organized as follows. Section~\ref{sec:sys} presents the system model as well as the sensing and training procedure, subsequently introducing the scheduling problem. We propose the novel scheduling framework in Section~\ref{sec:schem}. In Section~\ref{sec:results}, extensive numerical results and analyses are given, followed by the concluding remarks in Section~\ref{sec:con}.

\section{System Model} \label{sec:sys}
We consider a UAV swarm deployed over an area where they collect sensory data and collaboratively train a machine learning (ML) model for a given task (e.g., classification), as depicted in Fig.~\ref{fig:sys_model}.
Note that the UAVs are distributed in a way that they are dense enough to sense different parts of the region of interest.
We denote the set of UAVs as $\sV = \{ 1, \dots, V \}$, assuming that the UAVs hover at the same fixed altitude and maintain their 2-D positions during the mission time.
Note that the UAVs are deployed in an infrastructure-less setting, meaning that there is no ground station in charge of maneuvering them.
Because of this, each UAV is supposed to carry a battery that provides enough energy for the entire mission so that they do not have to return and recharge. 
\begin{figure}[htbp]
    \centering
    \includegraphics[width=1.0\columnwidth]{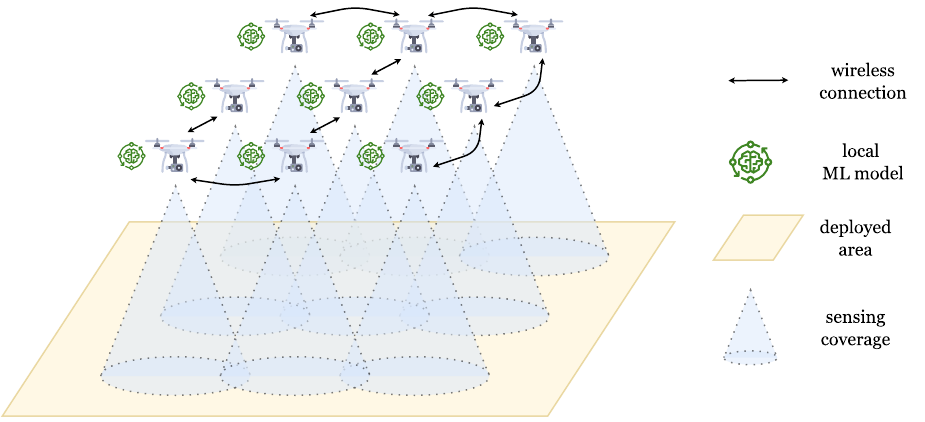}
    \caption{Example of UAV swarm deployment.}
    \label{fig:sys_model}
\end{figure}
\par

In terms of wireless communication, we assume that all UAVs have a maximum communication range $D^{\mathsf{com}}_{\max}$ to transmit and receive data with their neighbors.
All the UAVs are also supposed to be crowded enough so that every pair is able to communicate with each other either through direct or multi-hop communication. 
We assume that the available communication bandwidth is divided into subbands (e.g., using orthogonal frequency division multiple access or OFDMA), which are allocated to the UAVs so that concurrent transmissions between the UAVs are possible with negligible interference.
% Note that in actual scenarios, all the UAVs may not be connected and full connectivity can only be possible within each connected part of UAVs. However, we forgo this scenario in our work because our framework can be applied to each connected part in the network straightforwardly. 
\par

Recall that the mission of the UAVs is to collect the necessary data to train an ML model.
Because each UAV is dedicated to a sensing area, the model on each UAV would be limited to learning a handful of training information if it can only make use of its own available data.
Consequently, this could result in, for example, the poor predictive performance of the trained model at the end of the mission.
Therefore, it is evident that a collaborative online FL training process among the UAVs is necessary.
The following describes how the online FL training and aggregation procedure is conducted within the UAV swarm.

\subsection{Federated Learning in UAV Swarms} \label{sec:framework}
We consider an online FL framework for a UAV swarm as follows. Each UAV $v$ continuously collects its sensory data during the flight time. The local dataset on UAV $v$, denoted as $\mathcal{D}_v$ with $\left| \mathcal{D}_v \right|$ data samples, is used to train a model that aims to minimize an objective function $F_v (\vw)$, where $\vw$ is the model parameter vector. Then the global optimization problem in the whole swarm is formally defined as
\begin{equation}
\label{eq:glob_obj_func}
   \min_{\mathbf{w}}  F(\mathbf{w}) \triangleq \sum_{v \in \mathcal{V}} \lambda_v F_v(\mathbf{w}) ,
\end{equation}
where $\lambda_v = \left| \mathcal{D}_v \right| / \sum_{u \in \sV} \left| \mathcal{D}_u \right|$ is the weight coefficient corresponding to UAV $v$. To solve~\eqref{eq:glob_obj_func}, UAV $v$ performs an iterative stochastic gradient descent (SGD) update using its dataset $\mathcal{D}_v$ to obtain its local update $\vw^{(r)}_{v}$ at global round $r$. Then the global model update is calculated as
\begin{equation} \label{eq:weight_update}
\vw^{(r)} = \sum_{v \in \sV} \lambda_v \vw^{(r)}_v,
\end{equation}
which is then used as the starting training parameters for the next round $r+1$ on each UAV. It is seen that the global update~\eqref{eq:weight_update} depends on the contributions of many participating UAVs during the training process. Therefore, in the following, we go through the entire data collection and collaborative training process of the UAV swarm.

\subsection{Collaborative Sensing and Training Procedure}
In this part, we adopt the FL design including one aggregator and multiple training clients from the seminal FL design, which can take full advantage of distributed datasets~\cite{li_on_2020}.
Subsequently, the following steps describe the activities performed by the ordinary UAVs and a dedicated root UAV, denoted as $v_c$, deployed for a crowdsensing task.
The root UAV is responsible for aggregating all local models into a global model in each training round.
\begin{enumerate}[label=\protect\circled{\arabic*}]
    \item \textit{Initialization.} Each UAV transmits its beacon message notifying its surrounding neighbors of its presence and exchanging necessary network control information. At the end of this step, the UAVs know the information of the entire network and the root UAV $v_c$.
    
    \item \textit{Data collection.} The UAVs use their on-board sensors to acquire sensory data for a fixed period of time.
    
    \item \textit{Local training.} The UAVs perform local model updating based on their own dataset using the SGD method with a fixed number of iterations to obtain their $\vw^{(r)}_v$ at round $r$. 
    
    \item \textit{Model aggregation.} The root UAV $v_c$ performs the global model aggregation process by computing $\vw^{(r)}$ using~\eqref{eq:weight_update} based on the local updates obtained from other UAVs.
    
    \item \textit{Model dissemination.} The root UAV transmits $\vw^{(r)}$ to other UAVs. The whole UAV swarm is then ready for the next round $r+1$ starting from step~\circled{2}.
\end{enumerate}

Note that step~\circled{1} is executed once as the UAVs are deployed.
Steps~\circled{2}--\circled{5} are conducted within one global round and repeated until the training algorithm converges or the maximum number of global rounds is reached.
Next, we will provide in-depth analyses of delay and message transmission overhead during these steps.

\subsection{Delay and Overhead Analyses} \label{subsec:ana}
In step~\circled{1}, since the network discovery process is initiated only a single time and the network remains static within the context of our study, we can afford to overlook the delay and overhead associated with this step.
Next, in step~\circled{2}, all the UAVs collect the data in a fixed period of time, which is a predefined parameter.
Therefore, we can consider this to last a fixed duration with no transmission, thus leaving little room for delay reduction.
A similar explanation can also be used for step~\circled{3} as the number of SGD iterations is a selective parameter.
Suppose that all the UAVs have the same processing unit and conduct training in parallel, the training time spent by one UAV is also the amount for the whole network in this step.
\par

In step~\circled{4}, the UAVs have finished local training and their local updates have to be accumulated at $v_c$ to realize the global model aggregation process.
It is challenging to design an efficient transmission scheme for the delivery of local updates of individual UAVs to $v_c$.
This is because the large-scale UAV network requires multi-hop communications for exchanges of model parameters among UAVs, causing potentially large delays and transmission overhead.
We denote the time needed for $v_c$ to receive the local updates from the swarm as $T_{\mathsf{trans}}$ and the number of message transmissions among UAVs as $N_{\mathsf{trans}}$, which will be carefully revisited later.
Then after $T_{\mathsf{trans}}$, $v_c$ computes the global model, which is assumed to consume a negligible amount of time.
Finally, the model dissemination step~\circled{5} requires $v_c$ to send back the aggregated global model to other UAVs. This is a well-studied wireless broadcast problem and a known heuristic broadcast algorithm, e.g., \cite{pires_broadcast_2019}, can be employed to reduce both delays and transmission overhead.

\subsection{Scheduling Problem}
After the analyses, we are left with the delay $T_{\mathsf{trans}}$ and overhead $N_{\mathsf{trans}}$ in step~\circled{4} to be improved.
In fact, these two parameters are very much dependent on a scheduling scheme that manages the transmissions in the network.
Consider a toy example of transmission activities that may happen without a scheduler in this step as illustrated in Fig.~\ref{fig:naive1}.
In this example, each UAV $v_i$ individually transmits its local update $\vw_{v_i}$ to the root UAV $v_c$ for aggregation.
For simplicity, let us assume $\lambda_{v_i} = 1$ in~\eqref{eq:weight_update}. 
One can see that $\vw_{v_5}$ can traverse the path $(v_5, v_6, v_c)$ (i.e., the shortest path) instead of $(v_5, v_6, v_4, v_c)$ to save a transmission.
It can also be observed that we can reduce transmission overhead (or the required number of transmissions) if individual UAVs aggregate/accumulate the model parameters before forwarding them toward the root UAV.
In particular, if $v_1$ can add $\vw_{v_1}$, $\vw_{v_2}$ and $\vw_{v_3}$ together and send the sum in one message as depicted in Fig.~\ref{fig:naive2}, we can further reduce $N_{\mathsf{trans}}$. 

%At the same time, the delay $T_{\mathsf{trans}}$ can be reduced if $v_1$ can wait for $v_2$ and $v_3$ and aggregate them before transmitting the resulting weighted average model update in one message instead of transmitting them in separate messages.
\begin{figure} [htpb]
    \captionsetup[subfigure]{justification=centering}
    \centering
     \begin{subfigure}[b]{0.493\linewidth}
         \centering
         \includegraphics[width=\textwidth]{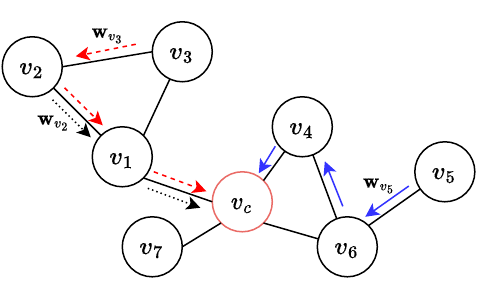}
         \caption{}
         \label{fig:naive1}
     \end{subfigure}
    \hfill
     \begin{subfigure}[b]{0.493\linewidth}
         \centering
         \includegraphics[width=\textwidth]{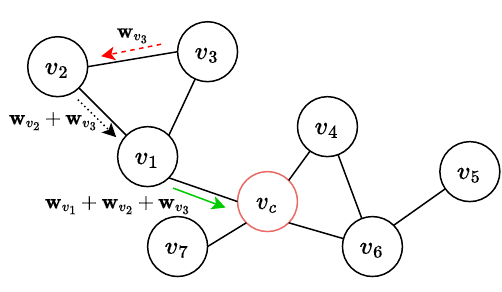}
         \caption{}
         \label{fig:naive2}
     \end{subfigure}
    \caption{Illustrations of transmission schemes: (a) Independent transmissions, (b) Combined aggregation and transmissions.}
    \label{fig:model_aggr}
\end{figure}

Based on the above intuitions and observations, in the next section, we design a scheme that aims to minimize both delay and overhead by efficiently scheduling transmissions during step~\circled{4} in a general case given any swarm topology. Since the UAV network is static during the entire flight time, the schedule can be applied to other global rounds until the termination conditions are met. 

\section{Proposed Scheduling Scheme} \label{sec:schem}
Our proposed scheduling scheme can be designed based on the following observations. 
If a node can perform ``local aggregation'' based on its own local update and that of other nodes, it can transmit this locally aggregated model as just one message instead of sending its own one and relaying others separately. 
However, each node must determine which local updates it should aggregate and thus wait to receive them. 
Therefore, a scheduling scheme is required for such coordination of model aggregation in the UAV network, which is the focus of our design.
In this section, based on the swarm deployment assumptions from Section~\ref{sec:sys}, we will formally address this through a graph-based approach.
\par

We consider the swarm as an undirected and unweighted connected graph $\sG = (\sV, \sE)$.
The set of UAVs $\sV$ is reused as the set of vertices or nodes\footnote{Hereinafter, ``UAV'' and ``node'' are used interchangeably.}, and $\sE$ is the set of edges formed as two nearby UAVs are within the direct communication range of each other.
We assume that $T_{\mathsf{trans}} = \eta \Delta t$ is divided into $\eta$ time slots and each time slot lasts a fixed duration $\Delta t$.
$\Delta t$ is also considered as the transmission time for a message between two adjacent UAVs.
Let $d(\cdot, \cdot)$ represent the number of hops along the shortest path for two UAVs.
Next, we aim to address the first crucial part in the process of designing the scheduling scheme, which is the choice of the root node.

\subsection{Root Node Selection and Associated Delay} \label{sec:root}
Because most of the transmission activities in the network are for the delivery of local updates to the root UAV $v_c$, here we specify how to choose an efficient one.
First, we define the node eccentricity of a node $v$ as $\epsilon(v) \triangleq \max_{u \in \sV} d(v, u)$, which is the longest distance in hop count from $v$ to any other node $u$ in the graph.
Then, a UAV $v_c \in \sV$ is chosen as the root node if $v_c = \argmin_{v \in \sV} \{ \max_v \{d(v, u)\}_{u \in \sV } \}$.
In other words, $v_c$ has the minimum eccentricity or the shortest longest distance from itself to other nodes.
We also define this value as the minimum graph eccentricity, denoted as $\epsilon(\sG) \triangleq \min \{ \max_v \{ d(u, v) \}_{u \in \sV_{\sG}} \}$.
\par 

The proposed choice of root node is based on the following observations. First, in terms of delay minimization, given any root node $u$, $u$ must wait for the furthest node to receive its local update. 
Formally, this amount of time can be calculated as $T_{\mathsf{trans}} = \max_{w \in \sV} d(u, w) \Delta t = \epsilon(u) \Delta t $. 
Then it is clear that to minimize $T_{\mathsf{trans}}$, we should choose the root node so that it has the shortest longest distance to any other node, which is exactly the definition of minimum graph eccentricity $\epsilon(\sG)$. 
Then we can formally state that, by the definition of $v_c$, $\min T_{\mathsf{trans}} = \epsilon(\sG) \Delta t = \epsilon(v_c) \Delta t \triangleq T^*$, i.e., $\min \eta = \epsilon(\sG)$.
This result aids in selecting the root node $v_c$ to obtain the minimum $T_{\mathsf{trans}}$ given an arbitrary swarm topology.
Next, we aim to analyze the communication overhead $N_{\mathsf{trans}}$.

\subsection{Communication Overhead Analysis}
Recall that in step~\circled{4}, each node has to transmit its local update to a root node $v_c$ for global model aggregation.
%and the root node does not have to transmit to any other node for aggregation.
Since an edge $e \in \sE$ is formed based on the direct communication of two nodes, we assume that each time a message traverses an edge, it can be counted as a transmission.
Then the number of messages to be transmitted is minimized when each node, except the root node, only transmits once to a direct neighbor.
Formally, we can state that $\min N_{\mathsf{trans}} = |\sV| - 1 \triangleq N^*$.
$N^*$ can be easily achieved in the setting where all nodes directly transmit their local update to an aggregator in single-hop networks.
However, the swarm forms a multi-hop network and thus, it requires multi-hop transmission for the model update delivery.
Also as previously discussed, if each node can deliver its own local update as well as that of the others in one transmission, the overhead can be reduced.
This intuition guides us to the next step, where we design a transmission scheduling scheme that aims to achieve both $T^*$ and $N^*$.

\subsection{Proposed Scheduling Scheme}
Our proposed scheduling scheme contains two main parts.
The first part, named \textproc{Scheduler}, is executed once by the root UAV $v_c$ at the end of step~\circled{1} after the UAV swarm has identified the root UAV $v_c$ chosen based on Section~\ref{sec:root}.
Then $v_c$ starts to construct a subgraph $\sI = (\sV, \sE_{\sI})$ as a breadth-first tree based on $\sG$, and a set $\sA$ of $\eta = \epsilon(\sI)$ transmission matrices $\mA$.
In detail, for every $t$ we define a zero-initialized transmission matrix $\mA^{(t)}$ of size $|\sV| \times |\sV|$.
We also define the tier $k$ of a node $v$ to be its distance (in hop count) $d(v, v_c)$ to the root $v_c$.
Subsequently, the set of all nodes at the same tier $k$ is denoted as $\sV_k$.
Then, a tier-$k$ node $i$ is scheduled to transmit to $j$ if $j$ belongs to the neighbor set $\sN_{\sI}(i)$ of $i$ in $\sI$ and is one-hop closer to $v_c$, i.e., $j \in \sV_{k-1}$, whose scheduled transmission is indicated by $\mA^{(t)}_{i,j} = 1$. 
The constructed $\sA$ is then sent to all other nodes.
\par

The second part is a function \textproc{NodeAction} that is executed in parallel by all $v \in \sV$ at every $t$ during~\circled{4}. At $t$, $v$ checks to receive other messages from its higher-tier neighbors (i.e., ones with larger $k$) in a set $\sR_{v}$, as scheduled in $\mA^{(t-1)}$ at time $t-1$.
As previously discussed, in our design, a node can combine its local update with that of others to transmit just one message denoted as $\vw^{\mathsf{aggr}}_v$. 
By combining these updates from $u \in \sR_v$ and its own local update $\vw_v$, node $v$ computes its $\vw^{\mathsf{aggr}}_v$.
Finally, node $v$ sends $\vw^{\mathsf{aggr}}_v$ to a lower-tier neighbor based on $\mA^{(t)}$.
The complete transmission scheduling algorithm is presented in \cref{algo:sched}. 

\begin{algorithm}
\caption{Proposed Transmission Scheduler} \label{algo:sched}
\small
\begin{algorithmic}[1]

\LeftComment{\textproc{Scheduler} is executed once by $v_c$ in step~\circled{1}.}
\Function{Scheduler}{$\sG$, $v_c$}
    \State Construct $\sI$ based on $\sG$ and $v_c$
    \State $t$, $\eta$, $\sA$ $\gets$ $1$, $\epsilon(\sI)$, $\{ \mA^{(1)}, \dots, \mA^{(t)}, \dots, \mA^{(\eta)} \}$
    % \State $\sA \gets \{ \mA^{(1)}, \dots, \mA^{(t)}, \dots, \mA^{(\eta)} \}$

    \While{$t \leq \eta$}
        \For{each $v$ in $\sV_{\eta - t + 1}$}
            \For{each $u$ in $\sN_{\sI}(v)$}
                \If{$u$ is in $\sV_{\eta - t}$}
                    \State $\mA^{(t)}_{v,u} \gets 1$
                \EndIf
            \EndFor
        \EndFor
        \State $t$ $\gets$ $t + 1$
    \EndWhile
    \State \Return $\sA$ to all other nodes
\EndFunction

\LeftComment{\textproc{NodeAction} is executed by all nodes $v \in \sV$ in parallel at every $t$ during step~\circled{4}.}
\Function{NodeAction}{$t$, $\sA$, $\sG$}
    \If{$t \neq \epsilon(\sG)- d(v, v_c) + 1$} \Return 
    \EndIf
    
    \State $\sR_{v}\gets \{\}$
    \For{row $u$ in $\mA^{(t-1)} \in \sA$}
        \If{$\mA^{(t-1)}_{u,v} == 1$}
            \State Add $u$ to $\sR_{v}$
        \EndIf
    \EndFor 
    \State Wait to receive all updates from nodes in $\sR_{v}$
    \State Perform aggregation: $\vw_v^{\mathsf{aggr}} \gets \lambda_{v}\vw_{v} + \sum_{u \in \sR_{v}} \vw_{u}^{\mathsf{aggr}}$

    \For{column $u$ in $\mA^{(t)} \in \sA$}
    \If{$\mA^{(t)}_{v, u} == 1$}
        \State Send $\vw_{v}^{\mathsf{aggr}}$ to $u$
        \State \Return
    \EndIf
    \EndFor 
\EndFunction
\end{algorithmic}
\end{algorithm}

For the proposed algorithm, we have the following results.
\begin{proposition} \label{the1}
At the end of global round $r$, the proposed transmission scheduling scheme in \cref{algo:sched} achieves:
\begin{enumerate}[label=\protect{\texttt{R\arabic*}}]
    \item A globally aggregated model equal to~\eqref{eq:weight_update}.
    \item The minimum transmission time $T^*$.
    \item The minimum number of transmitted messages $N^*$.
\end{enumerate}
\end{proposition}

\begin{proof}
The proof is deferred to Appendix~\ref{sec:theorem_proof}. 
\end{proof}

Fig.~\ref{fig:example} illustrates how the proposed scheduling scheme works for the swarm previously given in Fig.~\ref{fig:model_aggr}
%\footnote{For global aggregation at the root node, all weight coefficients $\lambda_v$ must be transmitted to the root node. However, the sizes of these coefficients are much smaller than that of model updates so we omit them in the analysis.}
In this example, the red edges belong to the subgraph $\sI$. At $t = 1$, tier-$2$ nodes in the set $\sV_2 = \{ v_2, v_3, v_5\}$ begin to transmit their updates to their respective neighbors in $\sV_1 = \{ v_1, v_4, v_6, v_7 \}$. At this point, they do not have any higher-tier neighbors so $\vw_{v_i}^{\mathsf{aggr}} = \vw_{v_i}$, $\forall v_i \in \sV_1$. At $t = 2$, as the nodes in $\sV_1$ receive the transmitted messages, they compute and transmit their $\vw_{v_j}^{\mathsf{aggr}}$, $v_j \in \sV_1$ based on the received messages and their own local updates $\vw_{v_j}$. Finally, the global model is $\lambda_{\vw_{v_c}} \vw_{v_c} + \sum_{v_j \in \sV_1} \vw_{v_j}^{\mathsf{aggr}}$ as computed by $v_c$. In this example, the delay $T_{\mathsf{trans}}$ is $\epsilon(\sI) \Delta t= 2 \Delta t$ and the overhead $N_{\mathsf{trans}}$ is $|\sV|-1 = 7$\;messages.

\begin{figure}[htbp]
    \centering
    \includegraphics[width=1\columnwidth]{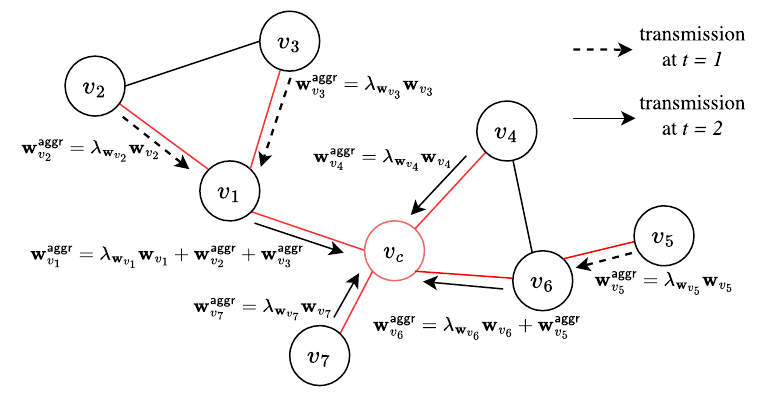}
    \caption{Example of proposed transmission scheduler.}
    \label{fig:example}
\end{figure}

\subsection{Complexity Analysis}
First, \textproc{Scheduler} begins with the construction of $\sI$ with a complexity of $\sO(|\sV|+|\sE_{\sG}|)$.
The computation of $\epsilon(\sI)$ is based on the breadth-first search algorithm and thus is $\sO(|\sV|^2+ |\sV|\times|\sE_{\sI}|)$.
Next, the initialization of $\sA$ can take $\sO(\eta \times |\sV|^2)$ and because $\eta$ can be bounded by $|\sV|$, this step may take $\sO(|\sV|^3)$.
Finally the nested loops also have a complexity of $\sO(|\sV|^3)$, leading to the complexity of the whole function being $\sO(|\sV|^3+ |\sV|\times|\sE_{\sI}|+|\sE_{\sG}|)$.
Regarding \textproc{NodeAction}, its initial condition and initialization steps both take $\sO(1)$.
The first loop, the aggregation step, and the second step all take $\sO(|\sV|)$ time, resulting in the complexity of $\sO(|\sV|)$ for the whole function.

\section{Numerical Results} \label{sec:results}
\subsection{Experiment Settings}
We randomly place different numbers of UAVs over an area of 1000\;m $\times$ 1000\;m so that the UAVs form a connected graph for given values of $D^{\mathsf{com}}_{\max} = 150$\;m and a safety distance of 5\;m. Note that all the UAVs participate in the training process of every global round. We evaluate our proposed scheme using two standard benchmark datasets, namely MNIST and CIFAR-10, for two classification tasks, whose training process aims to minimize a negative log-likelihood loss function. We also develop two simple convolutional neural networks mainly containing several convolutional and fully connected layers for the two tasks. Note that all the results in Section~\ref{sec:delay} are obtained by running each simulation as a global FL round five times and reporting the mean values.

\subsection{Convergence Results}
We show the convergence results of our proposed scheme and an FL approach for UAV swarms introduced by~\citet{qu_decentralized_2021} in terms of training loss and accuracy, as depicted in Figs.~\ref{fig:conv_mnist} and~\ref{fig:conv_cifar}.
In this scenario, we deploy 100 UAVs in total, and the batch size and learning rate are set to 10 and 0.01, respectively.
For both independent and identically distributed (IID) and non-IID settings, we report the training results until the global round where we observe the convergence. It can be seen that our proposed scheme can take advantage of the full participation of all UAVs similar to centralized FL, therefore achieving better convergence speed overall. 
\begin{figure} [htpb]
    \captionsetup[subfigure]{justification=centering}
    \centering
     \begin{subfigure}[b]{0.493\linewidth}
         \centering
         \includegraphics[width=\textwidth]{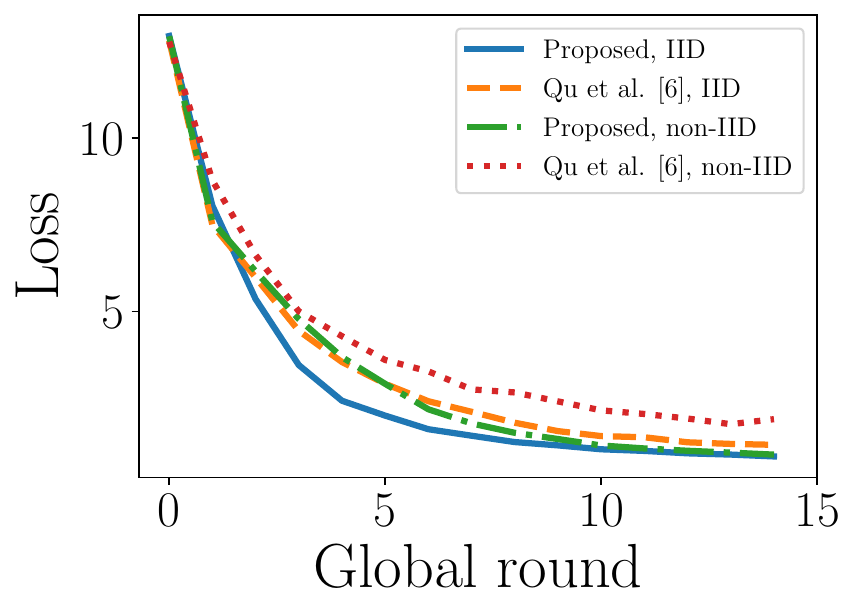}
         % \caption{}
         % \label{fig:loss}
     \end{subfigure}
    \hfill
     \begin{subfigure}[b]{0.493\linewidth}
         \centering
         \includegraphics[width=\textwidth]{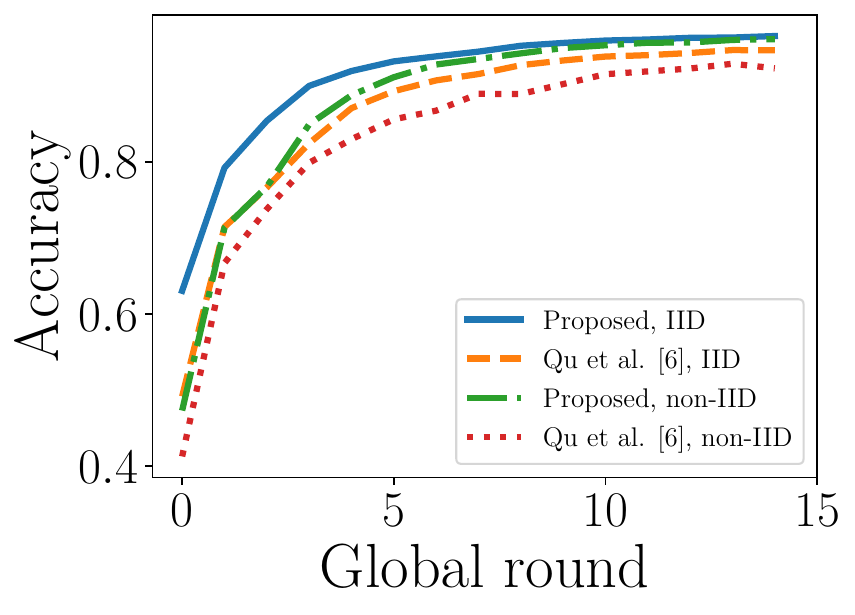}
         % \caption{}
         % \label{fig:acc}
     \end{subfigure}
    \caption{Loss and accuracy with MNIST dataset.}
    \label{fig:conv_mnist}
\end{figure}

\begin{figure} [htpb]
    \captionsetup[subfigure]{justification=centering}
    \centering
     \begin{subfigure}[b]{0.493\linewidth}
         \centering
         \includegraphics[width=\textwidth]{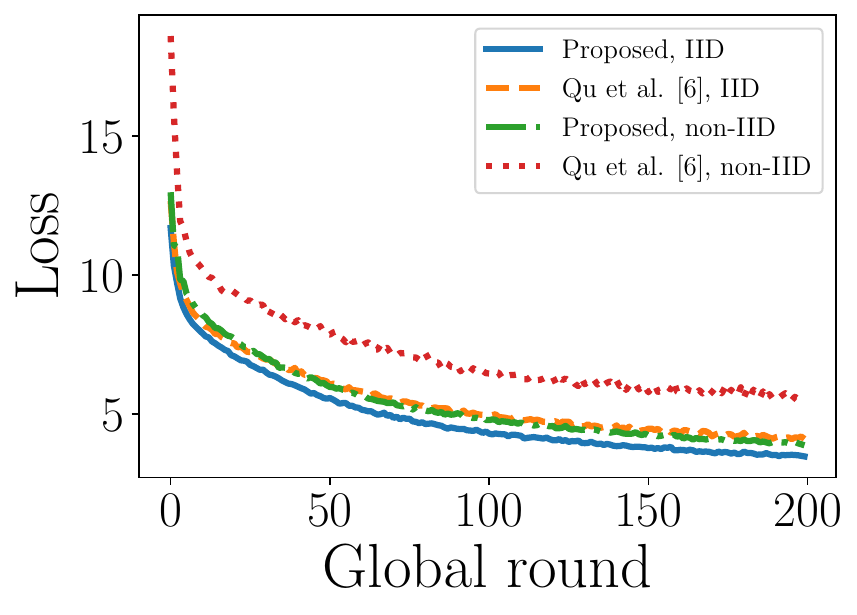}
         % \caption{}
         % \label{fig:loss}
     \end{subfigure}
    \hfill
     \begin{subfigure}[b]{0.493\linewidth}
         \centering
         \includegraphics[width=\textwidth]{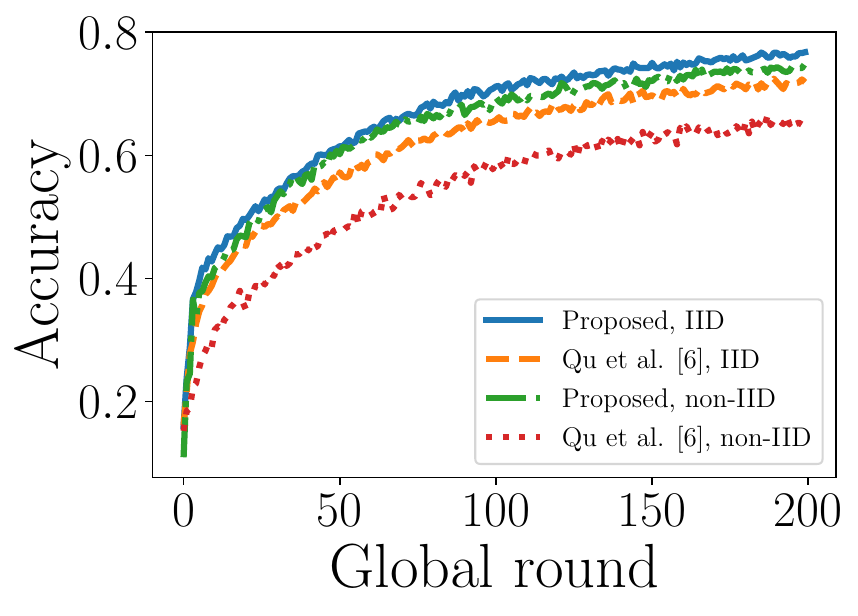}
         % \caption{}
         % \label{fig:acc}
     \end{subfigure}
    \caption{Loss and accuracy with CIFAR-10 dataset.}
    \label{fig:conv_cifar}
\end{figure}

\subsection{Delay and Overhead Results} \label{sec:delay}
In Fig.~\ref{fig:delay}, we compare the delay performance of different strategies of root node choice, which directly impacts the delay that the swarm needs to complete the model transmission process. To be specific, ``Random'' indicates the random selection of a UAV, ``Location-based'' is the choice of the node that is closest to central 2-D locations of the whole network, and ``Proposed'' indicates the selection as in \cref{algo:sched}. It is seen that the proposed scheme outperforms the two baseline schemes regardless of the size of the UAV swarm, especially at 60 and 80 UAVs. This is because we are always able to select the optimal root node given arbitrary swarm topologies and sizes.
\begin{figure}[htbp]
    \centering
    \includegraphics[width=0.5\columnwidth]{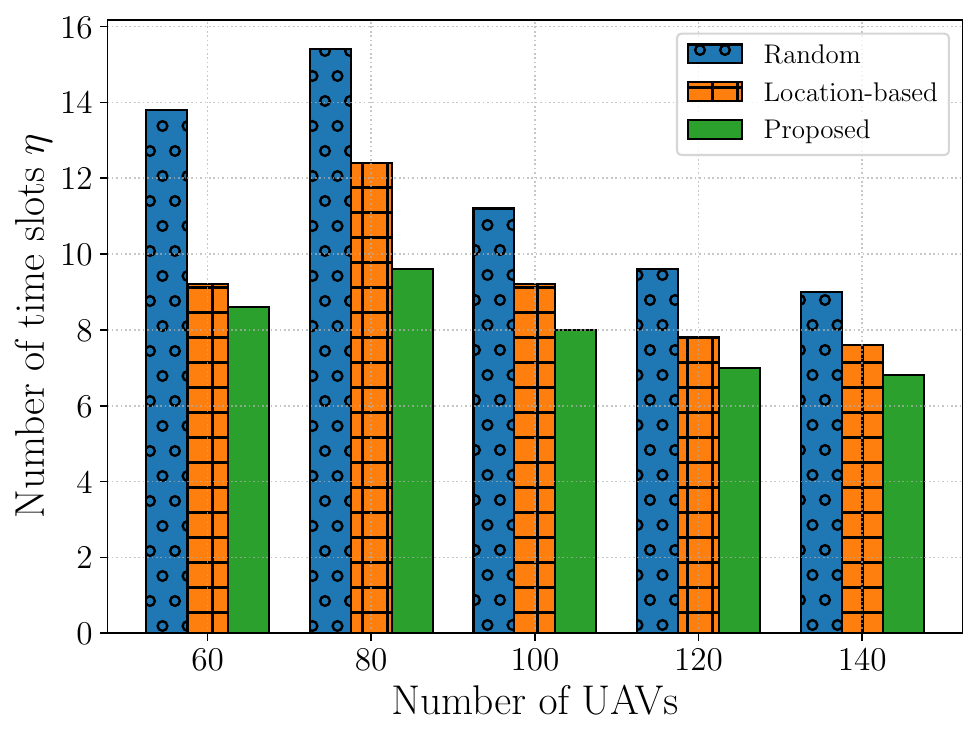}
    \caption{Number of time slots versus swarm size.}
    \label{fig:delay}
\end{figure}
\par

We further investigate the transmission overhead of model messages by different schemes against the swarm size in Fig.~\ref{fig:overhead}. In detail, ``Shortest Path First'' (SPF) denotes the shortest-path unicast of a local update by a node to the root node chosen in \cref{algo:sched}. Compared to the conventional SPF method, our proposed scheme can effectively save the number of transmissions for a wide range of swarm sizes. Meanwhile, it is seen that the amount of overhead by~\citet{qu_decentralized_2021} grows exponentially as the number of UAVs increases. Especially at 400 UAVs, the proposed scheme can save up to 71\% and 96\% of the traffic compared to SPF and~\cite{qu_decentralized_2021}, respectively. 

\begin{figure}[htbp]
    \centering
    \includegraphics[width=0.55\columnwidth]{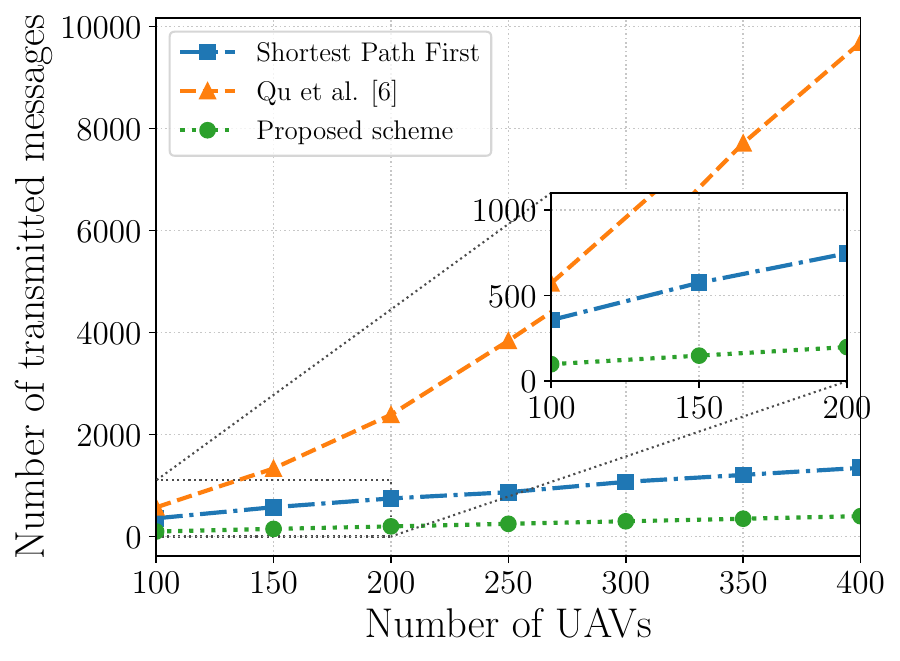}
    \caption{Amount of overhead versus swarm size.}
    \label{fig:overhead}
\end{figure}

\section{Conclusion} \label{sec:con}
We propose a transmission scheduling algorithm for FL conducted within UAV swarms.
To be specific, we first study the delay and communication overhead problems in the FL training process and then propose a transmission scheduler that achieves both minimum delay and overhead.
The simulation results conclude that the proposed scheme outperforms other baseline schemes in terms of convergence speed, delay, and overhead in the network.
Future work may include the study of how topology variation affects the FL framework.

\appendices
\section{Proof of \cref{the1}} \label{sec:theorem_proof}
% \begin{proof}

First, we introduce the following supplemental lemma.
\begin{lemma} \label{lem3}
$\sI$ is a minimum spanning tree, i.e., $ |\sE_{\sI}| = |\sV| - 1 $, and its minimum graph eccentricity is the same as $\sG$, i.e., $\epsilon(\sI) = \epsilon(\sG)$.
\end{lemma}
\begin{proof}\let\qed\relax
First, we prove that $\sI$ is a spanning tree of $\sG$ and, by recalling that $\sG$ is undirected and unweighted, any spanning tree is a minimum spanning tree. To be a spanning tree, $\sI$ must contain all the vertices and have no cycle. By \cite[Lemma~20.6]{cormen_introduction_2022}, applying a breadth-first search to an undirected and unweighted graph yields a breadth-first tree $\sI$ which includes a unique path from $v_c$ to any node in $\sG$ without cycles. Because of this, $\sI$ is also a spanning tree. Next, by applying \cite[Theorem~20.5]{cormen_introduction_2022} inductively to $\sI$, every unique path between node $v$ and $v_c$ in $\sI$ is a shortest path between them in $\sG$. It follows that $\epsilon(\sG) = \min\{ \max_v \{ d(u,v) \}_{u \in \sG} \} = \max_{v_c} \{ d(u,v_c) \}_{u \in \sV_{\sG}} = \max_{v_c} \{ d(u,v_c) \}_{u \in \sV_{\sI}} = \epsilon(\sI)$ since $\sV_{\sG} = \sV_{\sI} = \sV$ and $d(u, v_c)$ in $\sG$ is equal to $d(u, v_c)$ in $\sI$ for all $u \in \sV$.
\end{proof}
By \cref{lem3}, it follows that the transmission matrix $\mA^{(t)}$ at iteration $t$ in function \textproc{Scheduler} has exactly $|\sV_{\eta-t+1}|$ ($\eta = \epsilon(\sI)$) one-valued elements (or scheduled transmissions) because there is only one parent at tier $(\eta-t)$ for a tier-$(\eta-t+1)$ node.
Then, we prove \texttt{R1} as follows. At the end of step~\circled{4}, node $v_c$ performs aggregation using $\vw^{(r)} = \lambda_{v_c} \vw_{v_c} + \sum_{v \in \sV_1} \vw_{v}^{\mathsf{aggr}}$.
It is noted that for each $v \in \sV_1$, $\vw_{v}^{\mathsf{aggr}} = \lambda_v \vw_v + \sum_{u \in \sV_2 \cap \sN_{\sI}(v)} \vw_{u}^{\mathsf{aggr}}$.
By replacing each $\vw_{v}^{\mathsf{aggr}}$ recursively into $\vw^{(r)}$ and noting that $\sV = \cup_{k = 1}^{\epsilon(\sG)} \sV_k \cup \{v_c\}$, we obtain $\vw^{(r)}$ in~\eqref{eq:weight_update}.
Next, from \cref{lem3}, \cref{algo:sched} yields $T_{\mathsf{trans}} = \eta \Delta t = \epsilon(\sI) \Delta t = \epsilon(\sG) \Delta t = T^*$ by definition, which completes the proof for \texttt{R2}.
Finally, the proof for \texttt{R3} is given as follows. Given $[\cdot]$ being the Iverson bracket notation, the number of scheduled transmissions $N_{\mathsf{trans}} = \sum_{t = 1}^{\epsilon(\sG)} \sum_{i,j \in \sV} [\mA^{(t)}_{i, j} = 1] = \sum_{t = 1}^{\epsilon(\sG)} |\sV_t| = |\sE_{\sI}| = |\sV| - 1 = N^*$ also by \cref{lem3} and definition. 
\qed
% \end{proof}

\bibliographystyle{IEEEtranN}
{\footnotesize \bibliography{references}}

\end{document}